\documentclass{LMCS}

\def\dOi{11(2:15)2015}
\lmcsheading%
{\dOi}
{1--13}
{}
{}
{Jun.~15, 2014}
{Jun.~25, 2015}
{}

\ACMCCS{[{\bf Computing methodologies}]: Symbolic and algebraic
  manipulation--Symbolic and algebraic algorithms--Algebraic algorithms; 
  [{\bf Theory of computation}]: Computational complexity and
  cryptography--Complexity theory and logic\,/\,Algebraic complexity theory
}

\usepackage{hyperref}
\usepackage{amssymb}
\usepackage{amsmath}
\usepackage{amsthm}

\begin{document}

\title{On absorption in semigroups and $n$-ary semigroups}

\author[B.~Ba\v si\'c]{Bojan Ba\v si\'c}	
\address{Department of Mathematics and Informatics\\
University of Novi Sad\\
Trg Dositeja Obradovi\'ca 4\\
21000 Novi Sad\\
Serbia}	
\email{bojan.basic@dmi.uns.ac.rs}  





\begin{abstract}
  \noindent The notion of absorption was developed a few years ago by Barto and Kozik and immediately found many applications, particularly in topics related to the constraint satisfaction problem. We investigate the behavior of absorption in semigroups and $n$-ary semigroups (that is, algebras with one $n$-ary associative operation). In the case of semigroups, we give a simple necessary and sufficient condition for a semigroup to be absorbed by its subsemigroup. We then proceed to $n$-ary semigroups, where we conjecture an analogue of this necessary and sufficient condition, and prove that the conjectured condition is indeed necessary and sufficient for $\mathbf B$ to absorb $\mathbf A$ (where $\mathbf A$ is an $n$-ary semigroup and $\mathbf B$ is its $n$-ary subsemigroup) in the following three cases: when $\mathbf A$ is commutative, when $|A\setminus B|=1$ and when $\mathbf A$ is an idempotent ternary semigroup.
\end{abstract}

\maketitle

\section{Introduction}

Let $\mathbf A$ be an algebra and $\mathbf B\leqslant \mathbf A$. We say that $\mathbf B$ \emph{absorbs} $\mathbf A$, denoted by $\mathbf B\trianglelefteq \mathbf A$, iff there exists an idempotent term $t$ in $\mathbf A$ (that is, $t(a,a,\dots,a)\approx a$ for each $a\in A$) such that for each $a\in A$ and $b_1,b_2,\dots,b_m\in B$ we have
\[
\renewcommand{\arraystretch}{1.5}
\begin{array}{r@{}l}
t(a,b_2,b_3,\dots,b_m)&{}\in B;\\
t(b_1,a,b_3,\dots,b_m)&{}\in B;\\
\multicolumn{2}{c}{\vdots}\\
t(b_1,b_2,b_3,\dots,a)&{}\in B.
\end{array}
\]

The notion of absorption was developed a few years ago by Barto and Kozik, and immediately found many applications \cite{1,2,3,4,5}. We would particularly like to mention that Bulatov's dichotomy theorem for conservative CSPs \cite{bul}, with a deep and complicated proof (nearly 70 pages long), was reproved using these techniques on merely 10 pages \cite{6}. Loosely speaking, the main idea of absorption is that, when $\mathbf B\trianglelefteq \mathbf A$ where $\mathbf B$ is a proper subalgebra of $\mathbf A$, then some induction-like step can often be applied.

This naturally leads to the following question: given a finite algebra $\mathbf A$ and its subalgebra $\mathbf B$, is it decidable whether $\mathbf B\trianglelefteq \mathbf A$? This question turns out to be quite hard. Let us mention that the notion of absorption emerged as a generalization of the notion of the so-called near-unanimity term (in particular: an idempotent finite algebra $\mathbf A$ has a near-unanimity term iff every singleton absorbs $\mathbf A$). It was asked in 1995 whether the existence of a near-unanimity term in a finite algebra $\mathbf A$ is decidable \cite{mck}, and it took a while to finally prove that it is \cite{miklos} (another interesting point here is that, before this proof appeared, there were some evidences suggesting that the answer is actually negative). Some recent results on deciding absorption are given in \cite{k1,k2,k3}; as expected, the proposed algorithms are quite complex.

In this paper we show that in semigroups the absorption is much easier to grasp. Namely, for absorption in semigroups, in Theorem \ref{tt1} we provide a necessary and sufficient condition that is very easy to check. After that, we turn to $n$-ary semigroups, that is, algebras $\mathbf A=(A,f)$ where $f$ is an $n$-ary associative operation. We conjecture an analogue of the necessary and sufficient condition for $\mathbf{B}\trianglelefteq \mathbf{A}$ from Theorem \ref{tt1}, and we then prove the conjecture in the following cases: when $f$ is commutative, when $|A\setminus B|=1$ and when $\mathbf A$ is an idempotent ternary semigroup.

Let us say a few words on a possible application of these results. Namely, one of the most interesting algebraic results toward the CSP Dichotomy Conjecture of Feder and Vardi \cite{fv} is the proof that, if a finite relational structure $\Gamma$ does not admit any so-called weak near-unanimity (wnu) polymorphism, then $CSP(\Gamma)$ is $\textsf{NP}$-complete (see \cite{bjk}, where Bulatov, Jeavons and Krokhin gave a different algebraic sufficient condition for $CSP(\Gamma)$ to be $\textsf{NP}$-complete, and \cite{mm}, where Mar\'oti and McKenzie showed that this condition is equivalent to the nonexistence of a wnu polymorphism). Bulatov, Jeavons and Krokhin conjectured that the other direction also holds, that is, that the existence of a wnu polymorphism compatible with $\Gamma$ implies that $CSP(\Gamma)$ is in $\textsf P$ (this is known under the name Algebraic Dichotomy Conjecture). This has been checked for some relational structures of a special form, as well as for all relational structures but given the existence of a wnu polymorphism of a special form, and in all the cases known so far the results agree with the conjecture. In many of these works the absorption was the key ingredient in the proof (see, e.g., the references from the beginning of this section).

In particular, by the result of Jeavons, Cohen and Gyssens \cite{jcg}, we know that whenever $\Gamma$ admits a semilattice polymorphism (a semilattice operation is a binary operation that is idempotent, commutative and associative), then $CSP(\Gamma)$ is in $\textsf P$. Theorem \ref{tt1} from the present paper gives an exact description of when an algebra is absorbed by its subalgebra in the class of algebras with a binary associative operation (which is a wider class than the class of algebras with a semilattice operation). This provides a direct link between Theorem \ref{tt1} and the current line of attack on the Dichotomy Conjecture. Concerning our generalization to $n$-ary semigroups, so far there is no result (at least up to the author's knowledge) toward the Dichotomy Conjecture that directly relates to Conjecture \ref{conj} in a similar manner; however, since there are many results of this kind toward the Dichotomy Conjecture and many researchers are actively working on it, it is not hard to imagine that such a result exists and is just waiting to be discovered, and in fact, Conjecture \ref{conj} might serve as a motivation for it.

And of course, speaking about the notion of absorption itself, Theorem \ref{tt1} and Conjecture \ref{conj} may shed some light on the (presently quite unclear) behavior of absorption, since we now have a natural class of algebras in which the absorption behaves in a very predictable (but nontrivial) way. It might be a very useful research direction to discover whether there is a deeper reason for this nice behavior of absorption in semigroups and (conjecturally) $n$-ary semigroups, and whether this reason may help to describe the behavior of absorption in other classes of algebras.


\section{Absorption in semigroups}

The main (in fact, the only) theorem in this section is the following one.

\begin{thm}\label{tt1}
Let $\mathbf A=(A,\cdot)$ be a semigroup, and let $\mathbf B\leqslant\mathbf A$. Then $\mathbf B\trianglelefteq \mathbf A$ if and only if $ab\in B$ and $ba\in B$ for each $a\in A$, $b\in B$, and there exists a positive integer $k>1$ such that $a^k\approx a$ for each $a\in A$.
\end{thm}

\begin{proof}
($\Leftarrow$): Assume that the condition from the statement holds, and let us prove that $\mathbf B\trianglelefteq \mathbf A$. Choose a positive integer $k>1$ such that $a^k\approx a$ for each $a\in A$, and let $t(x,y)=x^{k-1}y$. For any $a\in A$ we have $t(a,a)=a^k\approx a$, that is, $t$ is an idempotent term. Further, for any $a\in A$, $b\in B$, we have $t(a,b)=a^{k-1}b=(a^{k-1})b\in B$ and $t(b,a)=b^{k-1}a=b(b^{k-2}a)\in B$, which proves that $t$ is an absorbing term.

($\Rightarrow$): 
Let $\mathbf B\trianglelefteq \mathbf A$, and let $t$ be an absorbing term. Since $t$ is an idempotent term, we trivially get that there exists a positive integer $k>1$ such that $a^k\approx a$ for each $a\in A$ (in particular, $k$ is the length of the term $t$). Therefore, we are left to prove that $ab\in B$ and $ba\in B$ for each $a\in A$, $b\in B$.

Let $t$ be a term in $m$ variables, which are named in such a way that the leftmost variable in $t(x_1,x_2,\dots,x_m)$ is $x_1$. Let $d_i$ denote the number of times the variable $x_i$ appears in $t(x_1,x_2,\dots,x_m)$.

Let $a\in A$ and $b\in B$ be given. We evaluate $t(ab,b^{k-1},b^{k-1},\dots,b^{k-1})$. Because $t(ab,b^{k-1},b^{k-1},\dots,b^{k-1})$ begins with $ab$ and $(ab)b^{k-1}=ab^k\approx ab$, we easily conclude
\[
t(ab,b^{k-1},b^{k-1},\dots,b^{k-1})\approx (ab)^{d_1}.
\]
Since $b^{k-1}\in B$ and $t$ is an absorbing term, by the previous equality we get
\begin{equation}\label{abe}
(ab)^{d_1}\in B.
\end{equation}

Let $r$ be any positive integer greater than $(m-1)d_1$ such that $r\equiv 1\pmod{k-1}$. Note that
\begin{equation}\label{cong}
rk\equiv 1\cdot 1=1\pmod{k-1}.
\end{equation}
For $2\leqslant i\leqslant m$, denote
\[
t_i=t((ab)^{d_1},b^{k-1},\dots,b^{k-1},(ab)^{r},b^{k-1},\dots,b^{k-1})
\]
(the expression $(ab)^{r}$ is at the $i^{\text{th}}$ coordinate and at all the coordinates denoted by ``$\dots$" we put $b^{k-1}$). Since $(ab)^{d_1}\in B$ (see (\ref{abe})), $b^{k-1}\in B$ and $t$ is an absorbing term, we get
\begin{equation}\label{tiB}
t_i\in B.
\end{equation}
As we have observed earlier, $(ab)b^{k-1}=ab^k\approx ab$, and we thus conclude that $t_i$ evaluates to a power of $ab$. In particular, since $(ab)^{d_1}$ appears $d_1$ times and $(ab)^{r}$ appears $d_i$ times, we obtain
\begin{equation}\label{tiizraz}
t_i\approx (ab)^{d_1^2+rd_i}.
\end{equation}

Consider the expression
\[
(ab)^{(r-(m-1)d_1)d_1}t_2t_3\cdots t_m.
\]
By (\ref{abe}) we get
\[
(ab)^{(r-(m-1)d_1)d_1}=((ab)^{d_1})^{r-(m-1)d_1}\in B, 
\]
which together with (\ref{tiB}) gives 
\begin{equation}\label{aaa}
(ab)^{(r-(m-1)d_1)d_1}t_2t_3\cdots t_m\in B.
\end{equation}
We further have
\[
\renewcommand{\arraystretch}{1.5}
\begin{array}{r@{}l}
(ab)^{(r-(m-1)d_1)d_1}t_2t_3\cdots t_m&{}\overset{\makebox[0pt][c]{\scriptsize(\ref{tiizraz})}}{\approx}(ab)^{(r-(m-1)d_1)d_1+\sum_{i=2}^m(d_1^2+rd_i)}\\
&{}=(ab)^{rd_1-(m-1)d_1^2+(m-1)d_1^2+\sum_{i=2}^mrd_i}\\
&{}=(ab)^{r(d_1+d_2+\cdots+d_m)}=(ab)^{rk}.
\end{array}
\]
Since $(ab)^k\approx ab$, it follows that $(ab)^{l_1}\approx (ab)^{l_2}$ whenever $l_1\equiv l_2\pmod{k-1}$. Therefore,
\[
(ab)^{(r-(m-1)d_1)d_1}t_2t_3\cdots t_m\approx (ab)^{rk}\overset{(\ref{cong})}{\approx}ab.
\]
Together with (\ref{aaa}), this gives $ab\in B$, which was to be proved. The proof that $ba\in B$ is analogous. This completes the proof of Theorem \ref{tt1}.
\end{proof}

\section{Absorption in $n$-ary semigroups}

We say that an $n$-ary operation $f:A^n\to A$ is \emph{associative} iff
\begin{equation}\label{nasoc}
\renewcommand{\arraystretch}{1.5}
\begin{array}{r@{}l}
f(f(a_1,a_2,\dots,a_n),a_{n+1},\dots,a_{2n-1})&{}=f(a_1,f(a_2,\dots,a_n,a_{n+1}),\dots,a_{2n-1})\\
&{}=\cdots\\
&{}=f(a_1,a_2,\dots,f(a_n,a_{n+1},\dots,a_{2n-1}))
\end{array}
\end{equation}
for every $a_1,a_2,\dots,a_{2n-1}\in A$. An algebra $\mathbf A=(A,f)$, where $f$ is an $n$-ary associative operation, is called an \emph{$n$-ary semigroup}.

Instead of $f(a_1,a_2,\dots,a_n)$ we shall often write $a_1a_2\cdots a_n$, instead of the expressions from (\ref{nasoc}) we shall write $a_1a_2\cdots a_{2n-1}$ etc. However, we have to keep in mind that such an expression, say $a_1a_2\cdots a_q$, is defined in $\mathbf A$ if and only if $q\equiv 1\pmod{n-1}$. On the other hand, any expression of the form $a_1a_2\cdots a_q$ (no matter whether $q\equiv 1\pmod{n-1}$ or not) will be called \emph{word}. Even if the word $w$ is not defined in $\mathbf A$, we shall still write $w^l$ for the concatenation $ww\cdots w$ (where $w$ is repeated $l$ times), but we need to be very careful not to apply any possible identities from $\mathbf A$ on such a word; for example, if $\mathbf A$ is an idempotent ternary semigroup and $a,b\in A$, then $(ab)^3a$ is a valid way to write $abababa$ (which is defined in $\mathbf A$), but we cannot deduce $(ab)^3a\approx aba$. The notation $w^l$ is defined also for $l=0$, and in that case it stands for the ``empty word", that is, $uw^0v$ means simply $uv$.

We believe that Theorem \ref{tt1} can be generalized for $n$-ary semigroups in the following way.

\begin{conj}\label{conj}
Let $\mathbf A=(A,f)$ be an $n$-ary semigroup, and let $\mathbf B\leqslant\mathbf A$. Then the following conditions are equivalent:
\begin{enumerate}
\item $\mathbf B\trianglelefteq \mathbf A$;
\item $ab^{n-1}\in B$ and $b^{n-1}a\in B$ for each $a\in A$, $b\in B$, and there exists a positive integer $k>1$ such that $a^k\approx a$ for each $a\in A$;
\item $a_1a_2\cdots a_n\in B$ whenever at least one of $a_1,a_2,\dots,a_n$ belongs to $B$, and there exists a positive integer $k>1$ such that $a^k\approx a$ for each $a\in A$. 
\end{enumerate}
\end{conj}

\noindent We say that an $n$-ary operation $f$ is \emph{commutative} iff
\[
f(a_1,a_2,\dots,a_n)=f(a_{\pi(1)},a_{\pi(2)},\dots,a_{\pi(n)})
\]
for any $a_1,a_2,\dots,a_n$ and any permutation $\pi$ of the set $\{1,2,\dots,n\}$. We now prove Conjecture \ref{conj} in the case when $f$ is commutative, and then in two more cases, namely when $|A\setminus B|=1$ and when $\mathbf A$ is an idempotent ternary semigroup.

\begin{thm}\label{t1}
Conjecture \ref{conj} holds when $f$ is commutative.
\end{thm}

\begin{proof}
The implications $\text{(2)}\Rightarrow\text{(3)}$ and $\text{(3)}\Rightarrow\text{(1)}$ are easy, and in fact we shall not use the commutativity of $f$ in their proofs. 

$\text{(2)}\Rightarrow\text{(3)}$: Let $a_1,a_2,\dots,a_n\in A$ be given, and let $a_i\in B$ for some $i$. Then
\[
a_1a_2\cdots a_i\cdots a_n\approx a_1a_2\cdots a_i^{2k-1}\cdots a_n=(a_1a_2\cdots a_i^{n-i+1}) a_i^{2k-n-2} (a_i^i\cdots a_n),
\]
which implies that it is enough to prove $a_1a_2\cdots a_i^{n-i+1}\in B$ and $a_i^i\cdots a_n\in B$. And indeed:
\[
a_1a_2\cdots a_i^{n-i+1}\approx a_1a_2\cdots a_i^{n-i+k}=(a_1a_2\cdots a_i^{k-i+1})a_i^{n-1}\in B
\]
and
\[
a_i^i\cdots a_n\approx a_i^{i+k-1}\cdots a_n=a_i^{n-1}(a_i^{i+k-n}\cdots a_n)\in B
\]
by the assumption.

$\text{(3)}\Rightarrow\text{(1)}$: Let $t(x,y)$ be any term of length $k$ containing at least one occurrence of each variable $x$ and $y$. Then $t$ is an absorbing term.

That leaves only the implication $\text{(1)}\Rightarrow\text{(2)}$. We also note that, since we have just shown that the previous two implications always hold, in the later theorems we prove only the implication $\text{(1)}\Rightarrow\text{(2)}$.

$\text{(1)}\Rightarrow\text{(2)}$: Let $\mathbf B\trianglelefteq \mathbf A$, and let $t(x_1,x_2,\dots,x_m)$ be an absorbing term. Let $k$ be the length of $t$. Then $a^k\approx a$ for each $a\in A$, and furthermore, $a^{l_1}\approx a^{l_2}$ whenever $l_1\equiv l_2\pmod{k-1}$.

By the commutativity of $f$, we may write $t(x_1,x_2,\dots,x_m)$ in the form $x_1^{k_1}x_2^{k_2}\cdots x_m^{k_m}$, where $k_1+k_2+\cdots+k_m=k$. Let us show that $b^{n-1}a\in B$ for each $a\in A$, $b\in B$.

Let $a\in A$, $b\in B$. Let
\[
t_1(a,b)=t(a,b,b,\dots,b) t(b,a,b,\dots,b) t(b,b,a,\dots,b)\cdots t(b,b,b,\dots,a) b^{l},
\]
where $l$ is chosen so that
\begin{equation}\label{KplusL}
m+l\equiv n\pmod{k-1}.
\end{equation}
The length of $t_1$ equals $mk+l$. Since $mk+l\equiv m+l\equiv n\pmod {k-1}$ and $n-1\mid k-1$, we get $mk+l\equiv n\equiv 1\pmod{n-1}$, that is, $t_1$ is well-defined. Further, since $t$ is an absorbing term, we have $t(a,b,b,\dots,b)\in B$, $t(b,a,b,\dots,b)\in B$, $\dots$, $t(b,b,b,\dots,a)\in B$, which gives $t_1(a,b)\in B$. Finally, note that
\begin{equation}\label{zaposle}
\renewcommand{\arraystretch}{1.5}
\begin{array}{r@{}l}
t_1(a,b)&{}\approx a^{k_1}b^{k-k_1}a^{k_2}b^{k-k_2}\cdots a^{k_m}b^{k-k_m}b^{l}\approx b^{mk-(k_1+k_2+\cdots+k_m)+l}a^{k_1+k_2+\cdots+k_m}\\
&{}\approx b^{(m-1)k+l}a^k\approx b^{m-1+l}a\overset{(\ref{KplusL})}{\approx} b^{n-1}a,
\end{array}
\end{equation}
which proves that $b^{n-1}a\in B$. The proof of $ab^{n-1}\in B$ is analogous.
\end{proof}

\begin{thm}
Conjecture \ref{conj} holds when $|A\setminus B|=1$.
\end{thm}

\begin{proof}
$\text{(1)}\Rightarrow\text{(2)}$: Let $\mathbf B\trianglelefteq \mathbf A$, and let $t(x_1,x_2,\dots,x_m)$ be an absorbing term. Let $k$ be the length of $t$. By the idempotence of $t$ it follows that $a^k\approx a$ and in fact that $a^{l_1}\approx a^{l_2}$ whenever $a\in A$ and $l_1$ and $l_2$ are positive integers such that $l_1\equiv l_2\pmod{k-1}$.

Let $A\setminus B=\{c\}$. All we have to prove is that $b^{n-1}c\not\approx c$ and $cb^{n-1}\not\approx c$ for each $b\in B$. In the first place, we shall prove that $c^{k-1}b\in B$ and $bc^{k-1}\in B$ for each $b\in B$, which shall be needed later.

Aiming for a contradiction, suppose first that $c^{k-1}b\approx c$ for some $b\in B$. Let $x_i$ be the leftmost variable in $t(x_1,x_2,\dots,x_m)$. Putting $c$ at the $i^{\text{th}}$ coordinate and $b$ at all the other ones gives
\[
t(b,\dots,b,c,b,\dots,b)=cc\cdots ccbb\cdots bbcc\cdots ccbb\cdots bb\cdots.
\]
Let us show that each occurrence of $cb$ in the above expression can be replaced by $cc$ without affecting the value of $t(b,\dots,b,c,b,\dots,b)$. And indeed, we have
\[
\cdots cb\cdots\approx\cdots c^kb\cdots=\cdots cc^{k-1}b\cdots\approx\cdots cc\cdots,
\]
which proves the claim. By iterating this process we ultimately get
\[
t(b,\dots,b,c,b,\dots,b)\approx c^k\approx c;
\]
however, since $t$ is an absorbing term, $t(b,\dots,b,c,b,\dots,b)\in B$ should hold, a contradiction. This proves that $c^{k-1}b\not\approx c$. In an analogous way we obtain that $bc^{k-1}\not\approx c$.

Again aiming for a contradiciton, suppose now that
\begin{equation}\label{aaa1}
b^{n-1}c\approx c
\end{equation}
for some $b\in B$. It now follows that
\[
c\approx c^k=ccc^{k-2}\overset{(\ref{aaa1})}{\approx} (b^{n-1}c)(b^{n-1}c)c^{k-2}=b(b^{n-2}cb)b^{n-3}(bc^{k-1}).
\]
Since $b\in B$ and $bc^{k-1}\in B$, it is impossible that $b^{n-2}cb\in B$, since then the value at the right-hand side would belong to $B$, while the value at the left-hand side is $c$. In other words,
\begin{equation}\label{aaa2}
b^{n-2}cb\approx c.
\end{equation}
Now, let
\[
t_1(c,b)=t(c,b,b,\dots,b) t(b,c,b,\dots,b) t(b,b,c,\dots,b)\cdots t(b,b,b,\dots,c) b^{l},
\]
where $l$ is chosen so that $m+l\equiv n\pmod {k-1}$. In exactly the same way as in the proof of Theorem \ref{t1}, we see that $t_1$ is well-defined and that $t_1(c,b)\in B$. Furthermore, we note that each occurrence of $cb$ in $t_1(c,b)$ can be replaced by $bc$ without affecting the value of $t_1(c,b)$; indeed:
\[
\cdots cb\cdots\overset{(\ref{aaa1})}{\approx} \cdots (b^{n-1}c)b\cdots= \cdots b(b^{n-2}cb)\cdots\overset{(\ref{aaa2})}{\approx}\cdots bc\cdots.
\]
This enables us to further mimic the proof of Theorem \ref{t1} (in particular, the lines (\ref{zaposle})), thus obtaining $t_1(c,b)\approx b^{n-1}c$ and hence $b^{n-1}c\in B$. However, this is exactly the opposite of the supposition (\ref{aaa1}). This condradiction proves $b^{n-1}c\in B$. The proof of $cb^{n-1}\in B$ is analogous.
\end{proof}

\begin{thm}\label{tern}
Conjecture \ref{conj} holds when $\mathbf A$ is an idempotent ternary semigroup.
\end{thm}

\begin{proof}
$\text{(1)}\Rightarrow\text{(2)}$: Let $\mathbf B\trianglelefteq \mathbf A$, and let $t(x_1,x_2,\dots,x_m)$ be an absorbing term, where the variables are named in such a way that the leftmost variable is $x_1$. Let $k$ be the length of $t$.

We need to prove that $ab^2\in B$ and $b^2a\in B$ for any $a\in A$, $b\in B$. The proof proceeds in nine steps:
\begin{enumerate}
\item We show that whenever $u^2vu\in B$, where $u$ is any word of an odd length and $v$ of an even length, then $vu\in B$. Analogously, we also obtain that, whenever $uvu^2\in B$, then $uv\in B$.
\item We show that whenever $ub\in B$, where $u$ is any word of an even length and $b\in B$, then $bu\in B$, and vice versa.
\item We show that $abbab\in B$ and $babba\in B$ for any $a\in A$, $b\in B$.
\item We show that $aab\in B$, $baa\in B$ and $aabaa\in B$ for any $a\in A$, $b\in B$.
\item We show that, whenever $t'(x,y)$ is a term such that $t'(a,b)\in B$ for all $a\in A$, $b\in B$, then $b(ab)^l\in B$, where $l$ is the absolute value of the difference of the number of occurrences of the letter $a$ at the odd, respectively even positions in the word $t'(a,b)$.
\item We show that, whenever $b(ab)^l\in B$ for an integer $l>1$ and some $a\in A$, $b\in B$, then $(ab)^{l-1}a\in B$.
\item We show that there exists a positive integer $l$ such that $b(ab)^l\in B$ and $b(ab)^{l+1}\in B$ for any $a\in A$, $b\in B$.
\item We show that $bab\in B$ for any $a\in A$, $b\in B$.
\item We show that $ab^2\in B$ and $b^2a\in B$ for any $a\in A$, $b\in B$.
\end{enumerate}\medskip

\noindent We now prove these steps.

\begin{enumerate}
\item Since $u^2vu\in B$, we obtain $t(vu,u^2vu,u^2vu,\dots,u^2vu)\in B$. Note that
\[
(vu)(u^2vu)\cdots=vu^3vu\cdots\approx vuvu\cdots=(vu)(vu)\cdots.
\]
Therefore,
\[
t(vu,u^2vu,u^2vu,\dots,u^2vu)\approx (vu)^k\approx vu,
\]
which gives $vu\in B$. The proof that $uvu^2\in B$ implies $uv\in B$ is analogous.
\item If $ub\in B$, then $bub^2=b(ub)b\in B$, and now $bu\in B$ by the previous step. The other direction is analogous.
\item Let $a\in A$, $b\in B$ be given. Denote
\[
a'=abbab. 
\]
We have $t(a'bb,b,b,\dots,b)\in B$. Note that $t(a'bb,b,b,\dots,b)$ is a word that starts with $a'$, ends with $b$ and has no occurrences of two letters $a'$ next to each other; furthermore, since $b^3\approx b$, this word reduces to a word that has either $b$ or $bb$ between each two successive occurrences of $a'$. We note
\[
a'ba'=(abbab)b(abbab)=(abb)^3ab\approx abbab=a',
\]
that is, each occurrence of $a'ba'$ in the considered word reduces to $a'$. It follows that $t(a'bb,b,b,\dots,b)$ reduces to a word $a'bba'bba'bba'\cdots b$, that is, either to $(a'bb)^l$ or to $(a'bb)^la'b$ for some odd positive integer $l$ ($l$ has to be odd for these products to be defined). These words further reduce to $a'bb$ and $a'bba'b$, respectively. Since
\[
a'bb=(abbab)bb=abbab^3\approx abbab
\]
and
\[
a'bba'b=(abbab)bb(abbab)b=abbab^3abbabb\approx abbababbabb,
\]
we obtain $abbab\in B$ or $abbababbabb\in B$. In the first case, this is what was to be proved. In the second case, since $(abb)ab(abb)^2=abbababbabb\in B$, by step 1 we again obtain $abbab\in B$. The proof that $babba\in B$ is analogous (or, alternatively, follows from $abbab\in B$ and step 2).

\item Let $a\in A$, $b\in B$ be given. By the previous step, we have
\[
aabaabb\approx aab^3aabb=(aab)bb(aab)b\in B.
\]
In an analogous way, we obtain
\[
bbaabaa\in B.
\]
From these two conclusion we get
\[
\renewcommand{\arraystretch}{1.5}
\begin{array}{r@{}l}
bbaab&{}\approx bb(aab)^3=bbaabaabaab\approx bbaaba^4baab^3\\
&{}=(bbaabaa)(aabaabb)b\in B.
\end{array}
\]
Now, since $bbaab\in B$, step 1 gives $aab\in B$. The proof that $baa\in B$ is analogous (or, alternatively, follows from $aab\in B$ and step 2). Finally, $aabaa\approx aab^3aa=(aab)b(baa)\in B$.

\item Let $a\in A$, $b\in B$ be given. We may assume that the leftmost variable in $t'(x,y)$ is $x$: indeed, if $t'(x,y)$ begins with $y$ and $t'(a,b)\in B$ for all $a\in A,b\in B$, then because of step 2 the same holds for the term obtained from the term $t'$ by moving the leftmost $y$ to the end, and this can be repeated until we reach a term that begins with $x$. We have $t'(abb,b)\in B$. The word $t'(abb,b)$ is a word that starts with $a$, ends with $b$, and has no occurrences of two letters $a$ next to each other. Since $b^3\approx b$, this word reduces to a word that has either $b$ or $bb$ between each two successive occurrences of $a$. We can write the obtained word in the form
\begin{equation}\label{sredj}
(ab)^{l_1}b(ab)^{l_2}b(ab)^{l_3}b\cdots
\end{equation}
with either $\cdots(ab)^{l_q}$ or $\cdots(ab)^{l_q}b$ at the end, for some positive integers $l_1,l_2,\dots,l_q$.

Given a word consisting only of the letters $a$ and $b$, by its \emph{difference} we shall mean the absolute value of the difference of the number of occurrences of the letter $a$ at the odd, respectively even positions in the considered word. We can assume that the difference of the word $t'(a,b)$ is nonzero, since otherwise the conclusion we have to reach is $b\in B$, which holds trivially.

Notice that the difference of the word $t'(a,b)$ equals the difference of the word $t'(abb,b)$ (the latter word is obtained by inserting $bb$ after each occurrence of $a$ in the former word, which keeps the parities of the positions at which $a$ appears in the former word), which in turn equals the difference of the word (\ref{sredj}) (replacing $b^3$ by $b$ also keeps the considered parities), which is evaluated to
\begin{equation}\label{diff}
|l_1-l_2+l_3-\cdots+(-1)^{q+1}l_q|.
\end{equation}

If $q=1$, then the word (\ref{sredj}) is $(ab)^{l_1}b$ (it cannot be $(ab)^{l_1}$ because this word has an even length) and its difference is $l_1$. Now from $(ab)^{l_1}b\in B$ and step 2 we obtain $b(ab)^{l_1}\in B$, which was to be proved.

Let $q=2$. Then the word (\ref{sredj}) is $(ab)^{l_1}b(ab)^{l_2}$ (there cannot be $\cdots(ab)^{l_2}b$ at the end, because then the length would be even). The difference of this word, which is assumed to be nonzero, equals $|l_1-l_2|$; therefore, $l_1\neq l_2$, and in particular, $l_1+l_2\geqslant 3$. Now, since $(ab)^{l_1}b(ab)^{l_2}\in B$ and $(ab)^{l_1-1}b(ab)^{l_2-1}\in A$, step 4 gives
\[
((ab)^{l_1-1}b(ab)^{l_2-1})^2((ab)^{l_1}b(ab)^{l_2})((ab)^{l_1-1}b(ab)^{l_2-1})^2\in B.
\]
Further, we have
\[
\renewcommand{\arraystretch}{1.5}
\begin{array}{l}
((ab)^{l_1-1}b(ab)^{l_2-1})^2((ab)^{l_1}b(ab)^{l_2})((ab)^{l_1-1}b(ab)^{l_2-1})^2\\
\hspace{2em}{}=(ab)^{l_1-1}b(ab)^{l_1+l_2-2}b(ab)^{l_1+l_2-1}b(ab)^{l_1+l_2-1}b(ab)^{l_1+l_2-2}b(ab)^{l_2-1}\\
\hspace{2em}{}=(ab)^{l_1-1}b((ab)^{l_1+l_2-2}bab)^3(ab)^{l_1+l_2-3}b(ab)^{l_2-1}\\
\hspace{2em}{}\approx(ab)^{l_1-1}b(ab)^{l_1+l_2-2}bab(ab)^{l_1+l_2-3}b(ab)^{l_2-1}\\
\hspace{2em}{}=(ab)^{l_1-1}b(ab)^{l_1+l_2-2}b(ab)^{l_1+l_2-2}b(ab)^{l_2-1}\\
\hspace{2em}{}=((ab)^{l_1-1}b(ab)^{l_2-1})^3\approx (ab)^{l_1-1}b(ab)^{l_2-1}.
\end{array}
\]
Therefore, if $(ab)^{l_1}b(ab)^{l_2}\in B$, then $(ab)^{l_1-1}b(ab)^{l_2-1}\in B$. Further, note that the difference of the latter word equals $|(l_1-1)-(l_2-1)|=|l_1-l_2|$, that is, this transformation preserves the difference. Repeatedly applying this procedure ultimately leads to $b(ab)^{l_2-l_1}\in B$ or $(ab)^{l_1-l_2}b\in B$ (depending on whether $l_2>l_1$ or $l_1>l_2$). In the former case, the claim is proved. In the latter case, step 2 gives $b(ab)^{l_1-l_2}\in B$, which again proves the claim.

Finally, let $q\geqslant 3$. We assume $l_1\leqslant l_q$. Let us explain why we are allowed to make this assumption. If it were $l_q<l_1$, then from $(ab)^{l_1}b(ab)^{l_2}b\cdots (ab)^{l_q}\in B$ (the case with $\cdots (ab)^{l_q}b$ at the end is similar) we get, by step 2, $(ba)^{l_1}b(ba)^{l_2}\cdots b(ba)^{l_q}\in B$, and now since $l_q<l_1$, everything that follows could be applied as if, informally speaking, it were read from right to left.

We claim that there exists a word that belongs to $B$, starts with $a$, ends with $b$, has no occurrences of two letters $a$ next to each other, has the same difference as the word (\ref{sredj}) and is shorter than (\ref{sredj}).

We first consider the case when $l_1\geqslant l_2$ or $l_1> l_3$. Since $l_1\leqslant l_q$, it follows that $l_2\leqslant l_q$ or $l_3< l_q$, respectively. We conclude that the minimal possible value of $l_i$, $1\leqslant i\leqslant q$, is achieved for at least one $i$ such that $2\leqslant i\leqslant q-1$. For such $i$ we have
\[
\renewcommand{\arraystretch}{1.5}
\begin{array}{l}
(ab)^{l_1}b(ab)^{l_2}b\cdots b(ab)^{l_{i-1}}b(ab)^{l_i}b(ab)^{l_{i+1}}b\cdots\\
\hspace{2em}{}=(ab)^{l_1}b(ab)^{l_2}b\cdots(ba)^{l_{i-1}-l_i}(b(ab)^{l_i})^3(ab)^{l_{i+1}-l_i}b\cdots\\
\hspace{2em}{}\approx(ab)^{l_1}b(ab)^{l_2}b\cdots(ba)^{l_{i-1}-l_i}b(ab)^{l_i}(ab)^{l_{i+1}-l_i}b\cdots\\
\hspace{2em}{}\approx(ab)^{l_1}b(ab)^{l_2}b\cdots b(ab)^{l_{i-1}+l_{i+1}-l_i}b\cdots	
\end{array}
\]
This gives the shorter word we were looking for. Indeed, it is enough to check that the obtained word has the same difference as the word (\ref{sredj}) (all the other requirements are immediately clear). And indeed, the difference of the obtained word equals
\[
\left|\sum_{j=1}^{i-2}(-1)^{j+1}l_j+(-1)^i(l_{i-1}+l_{i+1}-l_i)+\sum_{j=i+2}^q(-1)^{j+1}l_j\right|,
\]
which is easily seen to be equal to (\ref{diff}).

Let now $l_1<l_2$ and $l_1\leqslant l_3$. Since the word (\ref{sredj}) is in $B$ and $b(ab)^{l_2-l_1}\in A$, step 4 gives
\[
(b(ab)^{l_2-l_1})(b(ab)^{l_2-l_1})((ab)^{l_1}b(ab)^{l_2}b(ab)^{l_3}b\cdots)\in B.
\]
Further, we have
\[
\renewcommand{\arraystretch}{1.5}
\begin{array}{l}
(b(ab)^{l_2-l_1})(b(ab)^{l_2-l_1})((ab)^{l_1}b(ab)^{l_2}b(ab)^{l_3}b\cdots)\\
\hspace{2em}{}=b(ab)^{l_2-l_1}b(ab)^{l_2}b(ab)^{l_2}b(ab)^{l_3}b\cdots\\
\hspace{2em}{}=b((ab)^{l_2-l_1}b(ab)^{l_1})^3(ab)^{l_3-l_1}b\cdots\\
\hspace{2em}{}\approx b(ab)^{l_2-l_1}b(ab)^{l_1}(ab)^{l_3-l_1}b\cdots=b(ab)^{l_2-l_1}b(ab)^{l_3}b\cdots
\end{array}
\]
Now, since $b(ab)^{l_2-l_1}b(ab)^{l_3}b\cdots\in B$, by step 2 it follows that the word obtained from this word by moving the letter $b$ from the beginning to the end (and possibly applying $b^3\approx b$ at the end, in case that $b^3$ appears there) also belongs to $B$. This gives the shorter word we were looking for. Indeed, it is again enough to check only that the obtained word has the same difference as the word (\ref{sredj}), which follows by noting that the difference of the obtained word equals
\[
\left|(l_2-l_1)+\sum_{j=3}^q(-1)^jl_j\right|=\left|\sum_{j=1}^q(-1)^jl_j\right|=\left|\sum_{j=1}^q(-1)^{j+1}l_j\right|.
\]

To conclude the proof, we note that repeatedly applying this procedure of ``shortening the word" ultimately leads to a word of the form treated in one of the cases $q=1$ or $q=2$, from where we reach the desired conclusion in the already demonstrated way.

\item Let $b(ab)^l\in B$. By step 2, $(ab)^lb\in B$. Since $aab\in B$ (because of step 4), it follows that
\[
(ab)^laab=(ab)^{l-1}abaab\approx (ab)^{l-1}ab^3aab=((ab)^lb)b(aab)\in B.
\]
Denote $a'=(ab)^{l-1}a$ and $b'=(ab)^laab$. Since $b'\in B$, by step 4 we get $a'a'b'a'a'\in B$. Note that
\[
\renewcommand{\arraystretch}{1.5}
\begin{array}{r@{}l}
a'a'b'a'a'&{}=((ab)^{l-1}a)((ab)^{l-1}a)((ab)^laab)((ab)^{l-1}a)((ab)^{l-1}a)\\
&{}=(ab)^{l-1}a(ab)^{l-1}a(ab)^la(ab)^la(ab)^{l-1}a\\
&{}=(ab)^{l-1}a(ab)^{l-2}(aba(ab)^{l-1})^3a\\
&{}\approx(ab)^{l-1}a(ab)^{l-2}(aba(ab)^{l-1})a=((ab)^{l-1}a)^3\approx (ab)^{l-1}a.
\end{array}
\]
This completes the proof.

\item Let $o_i$, respectively $e_i$, denote the number of occurrences of the letter $x_i$ at the odd, respectively even, positions in the word $t(x_1,x_2,\dots,x_m)$ (we recall that $t$ is an absorbing term). Then the difference of the word $t(b,\dots,b,a,b,\dots,b)$ ($a$ is at the $i^{\text{th}}$ coordinate) equals $|o_i-e_i|$. We claim that these differences, for $1\leqslant i\leqslant m$, are coprime (not necessarily pairwise coprime). Suppose the opposite: there exists a prime number $p$ such that $p\mid |o_i-e_i|$ for each $i$, $1\leqslant i\leqslant m$. Then
\[
p\mid\sum_{i=1}^m(o_i-e_i)=\sum_{i=1}^mo_i-\sum_{i=1}^me_i=\left\lceil\frac k 2\right\rceil-\left\lfloor\frac k 2\right\rfloor=1
\]
(we recall that $k$ is the length of $t$), which is a contradiction. Therefore, the considered differences are coprime.

By step 5, if $l_i$ is any of these differences, then $b(ab)^{l_i}\in B$ for all $a\in A,b\in B$. If $l_i$ and $l_j$ are any two of these differences, then
\[
b(ab)^{l_i+l_j}=b(ab)^{l_i}(ab)^{l_j}\approx(b(ab)^{l_i})b(b(ab)^{l_j})\in B.
\]
Therefore, $b(ab)^l\in B$ whenever $l$ is any linear combination with nonnegative integer coefficients of the considered differences. Since these differences are coprime, any large enough positive integer can be represented as a linear combination of them. In particular, there indeed exists a positive integer $l$ such that $b(ab)^l\in B$ and $b(ab)^{l+1}\in B$.

\item Let $l_0$ be the least positive integer such that $b(ab)^{l_0}\in B$ and $b(ab)^{l_0+1}\in B$ for any $a\in A$, $b\in B$ (such a number exists by the previous step). We need to prove that $l_0=1$. Aiming for a contradiction, suppose that $l_0\geqslant 2$. We shall prove that $b(ab)^{l_0-1}\in B$ for any $a\in A$, $b\in B$, which, together with the assumed $b(ab)^{l_0}\in B$, contradicts the minimality of $l_0$.

We first treat the case $l_0=2$. Let $a\in A$, $b\in B$ be given, and let us prove that $bab\in B$. We have $b(ab)^2\in B$ and $b(ab)^3\in B$. By step 6, we have $aba\in B$ and $ababa\in B$. Since, by step 4, $aab\in B$ and $baa\in B$, we conclude
\[
(baa)(aba)(ababa)(aba)(baa)b(aab)\in B.
\]
Further, we have
\[
\renewcommand{\arraystretch}{1.5}
\begin{array}{r@{}l}
(baa)(aba)(ababa)(aba)(baa)b(aab)&{}=baa(abaab)^3aab\approx baaabaabaab\\
&{}=ba(aab)^3\approx baaab\approx bab.
\end{array}
\]
This proves the case $l_0=2$.

Let now $l_0\geqslant 3$. Let $a\in A$, $b\in B$ be given, and let us prove that $b(ab)^{l_0-1}\in B$. Denote $a'=bab$ and $b'=abbab$. By step 3, we have $b'\in B$. Therefore, $b'(a'b')^{l_0}\in B$ and $b'(a'b')^{l_0+1}\in B$. By step 6, we have $(a'b')^{l_0-1}a'\in B$ and $(a'b')^{l_0}a'\in B$. By step 4, we have $a'a'b'\in B$. Therefore,
\[
((a'b')^{l_0-1}a')((a'b')^{l_0}a')((a'b')^{l_0-1}a')b'(a'a'b')\in B.
\]
Further, we have
\[
\renewcommand{\arraystretch}{1.5}
\begin{array}{l}
((a'b')^{l_0-1}a')((a'b')^{l_0}a')((a'b')^{l_0-1}a')b'(a'a'b')\\
\hspace{2em}{}=(a'b')^{l_0-1}a'(a'b')^{l_0}a'(a'b')^{l_0}a'a'b'=((a'b')^{l_0-1}a'a'b')^3\\
\hspace{2em}{}\approx(a'b')^{l_0-1}a'a'b'=((bab)(abbab))^{l_0-1}(bab)(bab)(abbab)\\
\hspace{2em}{}=(bababbab)^{l_0-2}babab(bab)^3abbab\\
\hspace{2em}{}\approx (bababbab)^{l_0-2}bababbababbab=(bababbab)^{l_0-3}bababb(abbab)^3\\
\hspace{2em}{}\approx (bababbab)^{l_0-3}bababbabbab=ba((bab)^2ba)^{l_0-2}b\\
\hspace{2em}{}\approx ba(ba)^{l_0-2}b=b(ab)^{l_0-1}.\\
\end{array}
\]
(Between the last and the next to last row we used the fact that after each $(bab)^2$ there is another $bab$ following, and thus, because of $(bab)^3\approx bab$, we may simply erase each such $(bab)^2$.) This completes the proof.

\item Let $a\in A$, $b\in B$ be given. By the previous step, we have $bab\in B$. By step 2, we now obtain $ab^2\in B$ and $b^2a\in B$, which was to be proved.
\end{enumerate}

The proof of Theorem \ref{tern} is thus finished.
\end{proof}

\noindent For the end, we prove a proposition that shows that the requirement that $\mathbf A$ is idempotent from the previous theorem is, in a way, not so restrictive as it might seem to be.

\begin{prop}
Assume that Conjecture \ref{conj} holds for all idempotent $n$-ary semigroups. Then Conjecture \ref{conj} holds in general.
\end{prop}

\begin{proof}
As before, it is enough to prove only the direction $\text{(1)}\Rightarrow\text{(2)}$. Let $\mathbf B\trianglelefteq \mathbf A$, and let $t(x_1,x_2,\dots,x_m)$ be an absorbing term. Let $k$ be the length of $t$. Then the algebra $\mathbf A'=(A,f')$, where $f'(x_1,x_2,\dots,x_k)=x_1x_2\cdots x_k$, is a $k$-ary idempotent semigroup. The term $t$ is also an absorbing term for $\mathbf B$ in $\mathbf A'$, that is, $\mathbf B\trianglelefteq \mathbf A'$. Therefore, by the assumed special case of Conjecture \ref{conj}, $b^{k-1}a\in B$ and $ab^{k-1}\in B$ for each $a\in A$, $b\in B$. Of course, the same also holds in $\mathbf A$. From here it is easy to prove that $b^{n-1}a\in B$ and $ab^{n-1}\in B$ for each $a\in A$, $b\in B$; indeed:
\[
b^{n-1}a\approx b^{n+k-2}a=b^{n-1}(b^{k-1}a)\in B
\]
and an analogous reasoning shows $ab^{n-1}\in B$.
\end{proof}

\section*{Acknowledgments}
The author would like to thank the two anonymous referees for thorough reading of the paper and many useful comments, and in particular for a suggestion about a possible application of the results from the paper mentioned in the Introduction.

The research was supported by the Ministry of Science and Technological Development of Serbia (project 174006) and by the Provincial Secretariat for Science and Technological Development, Autonomous Province of Vojvodina (project ``Ordered structures and applications").



\begin{thebibliography}{10}
\bibitem{6}L. Barto, The dichotomy for conservative constraint satisfaction problems revisited, in: \emph{Proceedings of the 26th Annual IEEE Symposium on Logic in Computer Science (LICS 2011)}, Toronto, Ontario, Canada, 2011, pp. 301--310.
\bibitem{k3}L. Barto \& A. Kazda \& J. Bul\'in, The distance from congruence distributivity to near unanimity, \emph{General Algebra and Its Applications (GAIA2013)}, Melbourne, Australia, 2013.
\bibitem{3}L. Barto \& M. Kozik, Congruence distributivity implies bounded width, \emph{SIAM Journal on Computing} {\bf 39} (2010), 1531--1542.
\bibitem{4}L. Barto \& M. Kozik, Constraint satisfaction problems of bounded width, in: \emph{Proceedings of the 50th Symposium on Foundations of Computer Science (FOCS'09)}, 2009, pp. 595--603.
\bibitem{5}L. Barto \& M. Kozik, Absorbing subalgebras, cyclic terms, and the constraint satisfaction problem, \emph{Logical Methods in Computer Science} {\bf 8} (2012), 1--26.
\bibitem{1}L. Barto \& M. Kozik \& T. Niven, Graphs, polymorphisms and the complexity of homomorphism problems, in: \emph{Proceedings of the 40th annual ACM symposium on Theory of computing (STOC'08)}, New York, USA, 2008, pp. 789--796.
\bibitem{2}L. Barto \& M. Kozik \& T. Niven, The CSP dichotomy holds for digraphs with no sources and no sinks (a positive answer to a conjecture of Bang-Jensen and Hell), \emph{SIAM J. Comput.} {\bf 38} (2009), 1782--1802.
\bibitem{bul}A. A. Bulatov, Complexity of conservative constraint satisfaction problems, \emph{ACM Trans. Comput. Log.} {\bf 12} (2011), Art. 24, 66 pp.
\bibitem{bjk}A. Bulatov \& P. Jeavons \& A. Krokhin, Classifying the complexity of constraints using finite algebras, \emph{SIAM J. Comput.} {\bf 34} (2005), 720--742. 
\bibitem{k1}J. Bul\'in, Decidability of absorption in relational structures of bounded width, \emph{Algebra Universalis} {\bf 72} (2014), 15--28.
\bibitem{mck}B. A. Davey \& L. Heindorf \& R. McKenzie, Near unanimity: an obstacle to general duality theory, \emph{Algebra Universalis} {\bf 33} (1995), 428--439.
\bibitem{fv}T. Feder \& M. Y. Vardi, The computational structure of monotone monadic SNP and constraint satisfaction: a study through Datalog and group theory, \emph{SIAM J. Comput.} {\bf 28} (1999), 57--104.
\bibitem{jcg}P. Jeavons \& D. Cohen \& M. Gyssens, Closure properties of constraints, \emph{J. ACM} {\bf 44} (1997), 527--548. 
\bibitem{k2}A. Kazda, How to decide absorption, \emph{The 4th Novi Sad Algebraic Conference (NSAC2013)}, Novi Sad, Serbia, 2013.
\bibitem{miklos}M. Mar\'oti, The existence of a near-unanimity term in a finite algebra is decidable, \emph{J. Symbolic Logic} {\bf 74} (2009), 1001--1014.
\bibitem{mm}M. Mar\'oti \& R. McKenzie, Existence theorems for weakly symmetric operations, \emph{Algebra Universalis} {\bf 59} (2008), 463--489. 


\end{thebibliography}
\end{document}